\documentclass[conference,letterpaper]{IEEEtran}

\usepackage[utf8]{inputenc} 
\usepackage[T1]{fontenc}
\usepackage{url}
\usepackage{ifthen}
\usepackage{cite}
\usepackage[cmex10]{amsmath} 

\usepackage[
font=small,labelfont=bf
]{caption}
 \usepackage{flushend}
 \flushend
\usepackage{amsfonts, amssymb, amsthm, color,mathtools}
\usepackage{bbm}
\usepackage{comment,enumitem}

\usepackage{tikz}
\usepackage{diagbox}
\usetikzlibrary{shapes,arrows,positioning,calc,patterns,decorations}

\usepackage{booktabs, multirow}
\usepackage{tcolorbox}
\tcbuselibrary{theorems}

\newtheorem{defn}{Definition}
\newtheorem{theorem}{Theorem}
\newtheorem{lemma}{Lemma}
\newtheorem{con}{Construction}
\newtheorem{example}{Example}

\newcommand{\supp}{\text{supp}}
\newcommand{\mdef}{\ensuremath{\overset{\mathrm{def}}{=}}}

\renewcommand{\a}{\ensuremath{\mathbf{a}}}
\renewcommand{\c}{\ensuremath{\mathbf{c}}}
\renewcommand{\d}{\ensuremath{\mathbf{d}}}
\newcommand{\p}{\ensuremath{\mathbf{p}}}

\renewcommand{\u}{\ensuremath{\mathbf{u}}}
\newcommand{\x}{\ensuremath{\mathbf{x}}}
\newcommand{\y}{\ensuremath{\mathbf{y}}}
\newcommand{\z}{\ensuremath{\mathbf{z}}}
\newcommand{\w}{\ensuremath{\mathbf{w}}}

\newcommand{\C}{\ensuremath{\mathcal{C}}}
\newcommand{\D}{\ensuremath{\mathcal{D}}}
\newcommand{\Z}{\ensuremath{\mathbb{Z}}}

\newcommand{\0}{\ensuremath{\mathbf{0}}}
\newcommand{\1}{\ensuremath{\mathbf{1}}}

\IEEEoverridecommandlockouts

\begin{document}
\title{Optimal Codes Correcting a Burst of Deletions\\ of Variable Length} 


\author{%
  \IEEEauthorblockN{{\bf Andreas Lenz}, {\bf Nikita Polyanskii}}
  \IEEEauthorblockA{Technical University of Munich\\
                    Institute for Communications Engineering\\
                    DE-80333 Munich, Germany\\
                    Email: andreas.lenz@mytum.de, nikita.polianskii@tum.de}
\thanks{This work is funded by the European Research Council under the EU’s Horizon 2020 research and innovation programme (grant No. 801434).}
}


\maketitle

\begin{abstract}

	In this paper, we present an efficiently encodable and decodable code construction that is capable of correction a burst of deletions of length at most $k$. The redundancy of this code is $\log n + k(k+1)/2\log \log n+c_k$ for some constant $c_k$ that only depends on $k$ and thus is scaling-optimal. The code can be split into two main components. First, we impose a constraint that allows to locate the burst of deletions up to an interval of size roughly $\log n$. Then, with the knowledge of the approximate location of the burst, we use several {shifted Varshamov-Tenengolts} codes to correct the burst of deletions, which only requires a small amount of redundancy since the location is already known up to an interval of small size. Finally, we show how to efficiently encode and decode the code.
\end{abstract}

\section{Introduction}
Burst deletions and insertions are a class of errors that can be found in a variety of applications, ranging from modern data storage systems, e.g., DNA-based data storage over communication systems to file synchronization. In contrast to classical deletion and insertions errors, that delete and insert symbols into a string at arbitrary positions, burst errors occur at consecutive positions. 

The study of classical deletion correcting codes goes back to the work of Levenshtein \cite{Lev65}, where he established that any code that corrects $k$ deletions if and only if it can correct $k$ insertions and deletions and showed that the {Varshamov-Tenengolts (VT) codes} \cite{varshamov_1965} are capable of correcting a single insertion or deletion. The VT codes have later been extended to non-binary alphabets \cite{tenengolts_nonbinary_1984} and also to the case of multiple deletions \cite{helberg_multiple_2002}. However, the latter construction already exhibits a redundancy that is linear in $n$ for any $k \geq 2$. Motivated by the fact that markers inside codewords allow for an easier detection of deletion errors Brakensiek et al. \cite{brakensiek2017efficient} used implicit markers, respectively \emph{patterns}, as boundaries for the symbols of several outer codes, resulting in a code of redundancy $O(k^2 \log k \log n)$. More recently, Sima et al. \cite{sima2019optimal} refined this technique by protecting a pattern-indicator vector from errors. By exploiting the fact that this vector is sparse, i.e., contains few $1's$, and using a single outer code, they were able to construct a code of redundancy $O(k \log n)$. In the context of document exchange Cheng et al. \cite{cheng_deterministic_2018} proposed another family of $k$-deletion-correcting codes with redundancy $O(k \log n)$, and in \cite{haeupler_optimal_2019} a $k$-deletion-correcting code with redundancy $O(k \log^2(n/k))$ has been found. Further works focus on the case, where the number of deletions scale with the codeword length \cite{guruswami_polynomial_2019} and on larger alphabet sizes \cite{haeupler_synchronization_2017}. For a detailed review of deletion-correcting codes, we refer to the surveys of Mitzenmacher \cite{mitzenmacher_survey_2008} and Mercier \cite{mercier_survey_2010}.

A related type of errors are tandem duplications, an object of study in the field of biological information theory. Here, a block of $k$ symbols inside a string are duplicated and inserted right after the occurrence of the original block. They have been studied in different contexts, including entropy rates \cite{farnoud_capacity_2016}, zero-error capacity \cite{kovacevic_zero-error_2019}, and error-correcting codes \cite{JFSB16, lenz_duplication-correcting_2019, kovacevic_asymptotically_2018}. Clearly, tandem duplications form a special sort of burst insertion errors, however, since burst insertions are a more general type of error, it is not possible to use tandem-duplication-correcting codes for burst insertions.


The concept of burst-deletion-correcting codes has been introduced by Levenshtein~\cite{levenshtein1967asymptotically}, where he presented a construction that can correct a burst of deletions of size at most $2$ with optimal redundancy\footnote{We say that the redundancy of a code construction is optimal if, for a fixed burst-length $k$, the ratio between the redundancy of the construction and some lower bound, e.g., sphere-packing bound, on the redundancy approaches $1$ as the code length goes to infinity.}. Note that there is an important difference between codes that can correct a burst of length \emph{at most} $k$ and a burst of length \emph{exactly} $k$, as a code of the earlier type can correct errors of the latter, but the converse is not true in general. For the latter, in \cite{cheng_codes_2014}, a construction with redundancy $k(\log(n/k+1))$ has been found. This construction has been improved in \cite{schoeny2017codes} to an optimal redundancy of $\log(n) + (k - 1) \log(\log(n)) + k - \log(k)$, while its non-binary generalization has been discussed in \cite{saeki_improvement_2018}. In the same paper \cite{schoeny2017codes}, a code using several VT constraints that can correct a burst of at most $k$ deletions with redundancy $(k-1)\log n + \left(k(k+1)/2-1\right) \log \log n +c_k'$, for some constant $c_k'$ that only depends on $k$, has been presented. The construction from \cite{schoeny2017codes} has been further improved to a redundancy of $\lceil\log k \rceil\log n +\left(k(k+1)/2-1\right) \log \log n +c_k''$ in~\cite{gabrys2017codes} by reusing some of the VT constraints. However, a sphere-packing argument shows that a redundancy of only roughly at least $\log n$ is required and thus the intriguing question of finding optimal codes for the case, where the length of the burst is at most $k$, still remained open.

The main contribution of this paper is deriving a construction of codes that can correct a burst of deletions of length at most $k$ with optimal redundancy
$
\log n + \frac{k(k+1)}{2} \log\log n+ c_k.
$
for a constant $c_k$ that only depends on $k$. Note that the value of $c_k$ will be given explicitly in the proof of Theorem \ref{th::one-burst-deletion code}. The remainder of the paper is organized as follows. We introduce notations and important definitions in Section~\ref{s::preliminaries}. An optimal code correcting one burst of deletions is presented in Section~\ref{s::simple construction}. Efficient encoding and decoding algorithms for another optimal construction are provided in Section~\ref{s::efficient construction}.
\section{Preliminaries}\label{s::preliminaries}
We start by introducing some notation that is used throughout the paper. Let $\mathbb{Z}$ be the set of integers and $[n]$ be the set of integers from $1$ to $n$ and \mbox{$[i,j]_\Delta \mdef \{i,i+\Delta,\dots,i+\lfloor\frac{j-i}{\Delta}\rfloor\Delta\}$} be the set of integers from $i$ to $j$ in intervals of $\Delta$, where we abbreviate $[i,j] \mdef [i,j]_1$.  For a string $\x = (x_1,x_2,\dots,x_n)$ of length $|\x|=n$, we write $\x_{[i,j]_\Delta} =(x_i,x_{i+\Delta}, x_{i+2\Delta},\dots, x_{i+\lfloor\frac{j-i}{\Delta}\rfloor\Delta})$ as the subsequence that contains all symbols from $\x$ in $[i,j]_\Delta$. We let $(\x,\y)$ denote the concatenation of two strings $\x$ and $\y$. A \emph{run} in a string $\x$ is a maximal substring, which is a consecutive repetition of one letter. Let $\1^i$ and $\0^j$ denote strings of $i$ consecutive ones and $j$ consecutive zeros, respectively. We write $\log x$ to denote the logarithm of $x$ in base two. We proceed with a rigorous definition of a burst of deletions and corresponding error-correcting codes.

\begin{defn} Let $k \in \mathbb{N}$. We define $B_k$(\x) to be the set of all possible
	outputs through the $k$-burst deletion channel, i.e.,
	$$
	B_k(\x)=\left\{\left(\x_{[1,i-1]},\x_{[i+k,n]}\right):\,i\in [n-k+1]\right\}.
	$$
	We further define $B_{\le k}(\x)$ to be the union of $B_{k'}(\x)$ for all $k'$ from $1$ to $k$.
\end{defn}
Throughout the paper we will denote by $k$ the maximum length of the burst of deletions and by $k'$ the length of the burst that actually happened. 
\begin{defn}
	Let $\C\subset \{0,1\}^n$ be a code of length $n$. The code $\C$ is said to be $(\le k)$-burst-deletion-correcting if for all $\x,\y\in\C$ with $\x \neq \y$, $B_{\le k}(\x) \cap B_{\le k}(\y) = \emptyset$.
\end{defn}
Note that it has been shown \cite{gabrys2017codes} that a code can correct a burst of deletions if and only if it can correct a burst of insertions, for both, the case where the length of the burst is fixed and variable, and we thus focus on the case of deletions here. We now turn to introducing concepts that will be required for the code construction.
\begin{defn}
	Let $\x$ and $\p$ be binary strings of length $n$ and $m$, respectively. Then, we define the indicator vector of the pattern $\p$ in $\x$ to be a vector of length $n$ with entries
	$$ \mathbbm{1}_\p(\x)_i = 
	\begin{cases}
	1, \quad \text{if } \x_{[i,i+m-1]} = \p \text{ and } i\leq n-m+1, \\
	0, \quad \text{otherwise}
	\end{cases}
	$$
	Further let $n_{\p}(\x)$ be the number of ones in $\mathbbm{1}_\p(\x)$. We define $\a_\p(\x)$ to be a vector of length $n_\p(\x)+1$ whose $i$-th entry is the distance between positions of the $i$-th and $(i+1)$-st $1$ in the string $(1,\mathbbm{1}_\p(\x),1)$.
\end{defn}
Note that this definition, the elements of the vector $\a_\p(\x)$ sum up to $n+1$. This allows to define the notion of pattern-dense strings as follows.
\begin{defn}
	Let $\p \in \{0,1\}^m$ and $\delta >m$ be a positive integer. A string $\x$ is called $(\p,\delta)$-dense, if each interval of length $\delta$ in $\x$ contains at least one pattern $\p$, i.e., for each $i\in [n-\delta+1]$ there exists $j \in [i,i+\delta-m]$ such that $\p = \x_{[j,j+m-1]}$. 
\end{defn}
We conclude with standard definitions of the Varshamov-Tenengolts checksum and parity check checksum.
\begin{defn}
	Given a string $\a\in \Z^n$, we define the Varshamov-Tenengolts and parity checksum be defined by
	$$
	\mathsf{VT}(\a)=\sum_{i=1}^{n}ia_i,\quad \mathsf{P}(\a)=\sum_{i=1}^{n}a_i.
	$$
\end{defn}

\section{Optimal code correcting one burst of deletions of length at most $k$}\label{s::simple construction}
Before we start introducing a construction of an optimal code, we give a high-level overview of our construction. The code consists of two components. 
The first component is presented in Sections \ref{ss::properties of dense strings} and \ref{subsec:locate} and ensures that the burst of deletions can be \emph{located} up to a small interval of size at most $\delta$. This will be achieved by fixing a certain pattern $\p$ and densifying the codeword, such that this pattern occurs at least once in each small interval of length $\delta$. These patterns will be used to define an indicator vector, similar to the construction in \cite{sima2019optimal}, that will be protected using a modified VT code, which is defined over integers, similar to that in \cite{levenshtein1967asymptotically}. The second component then uses the fact that the burst is known up to an interval of length $\delta$ and consists of shifted VT codes, that efficiently allow to correct deletions, given their approximate location, as will be shown in Section \ref{subsec:svt}. The final construction is presented in Section \ref{ss::code construction}. Throughout the paper, we fix the pattern to be $\p = (\0^k \1^k)$ and the density to $\delta = k2^{2k+1} \lceil \log n \rceil$. 
Note that we choose the pattern $\p$ to be \emph{resistant} to burst deletions in the sense that the numbers of patterns in an original string $\x$ and the erroneous string $\y \in B_{\leq k}(\x)$ satisfy $ -1\leq n_\p(\x)-n_\p(\y) \leq 2$. This is clearly not the case for arbitrary patterns $\p$ and since this will be helpful for the code construction, we restrict ourselves to patterns of the above mentioned form.
\subsection{Properties of $(\p,\delta)$-dense strings}\label{ss::properties of dense strings}

By the definition of $\D_{\p,\delta}$, we note three trivial properties of any $\x\in\D_{\p,\delta}$ in the following.
\begin{enumerate}
	\item Every element of $\a_\p(\x)$ is at most $\delta$.
	\item Every element $a_\p(\x)_i$, where $i\geq 2$, is at least $2k$~\footnote{Note that it is possible that the first element becomes $a_\p(\x)_1 < 2k$, if the first appearance of the pattern $\p$ in $\x$ is within the first $2k$ positions.}.
	\item The number of patterns in $\p$ is at most $n_\p(\x) \leq \frac{|\x|}{2k}$.
\end{enumerate}
Additionally, we prove a statement saying that almost all strings are $(\p,\delta)$-dense, which allows an encoding into $(\p.\delta)$ strings with little redundancy, as we will show later.
\begin{lemma}\label{lem::cardinality of dense family}
	For any $n\geq 5$, the number of $(\p,\delta)$-dense strings of length $n$ is at least
	$$
|	\{0,1\}^n \cap \D_{\p,\delta}| \geq 2^n (1-n^{1-\log e})\ge  2^{n-1}.
	$$
\end{lemma}
\begin{proof}[Proof of Lemma~\ref{lem::cardinality of dense family}]
	Let $\z$ be a random string chosen uniformly from the set $\{0,1\}^n$ and $E_i$ be the event that $\z_{[i+1,i+\delta]}$ doesn't contain the pattern $\p$. The probability of $E_i$ is at most 
	\begin{align*}
		\Pr(E_i) &\leq \prod_{j=0}^{\frac{\delta}{2k}-1} \Pr(\z_{[i+2kj+1,i+2k(j+1)]}\neq \p) =\left(\frac{2^{2k}-1}{2^{2k}}\right)^{\frac{\delta}{2k}} \\&= \left(1-1/2^{2k}\right)^{2^{2k}\lceil\log n\rceil} \le e^{-\log n}= \frac{1}{n^{\log e}}.
	\end{align*}
	Here, we use the property that $(1-1/x)^x\le e^{-1}$ for $x\ge 1$.
	Therefore, we have by the union bound that probability of the event that $\z$ is not in $\D_{\p,\delta}$ is upper bounded by
	$$
	\Pr(\z\not\in\D_{\p,\delta})\le (n-\delta+1)\Pr(E_i)\le \frac{1}{n^{\log e-1}}\le \frac{1}{2}, 
	$$ 
	where in the final inequality we used that $n \geq 5$.
\end{proof}
\subsection{Locating the burst of deletions} \label{subsec:locate}
In this subsection, we show how to construct a code that allows to locate the burst of deletions up to an uncertainty of $\delta$. The following construction combines the previously introduced $(\p,\delta)$-dense strings together with a VT-code.
\begin{con}
	For any integers $c_0$ and $c_1$, let 
	$$\C_\mathsf{loc}(c_0,c_1) \hspace{-.08cm}=\hspace{-.08cm} \left\{\x\in \{0,1\}^n: \hspace{-.08cm}\begin{array}{ll} \x \in \D_{\p,\delta},& \\n_\p(\x)= c_0 &(\bmod~4), \\\mathsf{VT}(\a_\p(\x)) \hspace{-.08cm}=\hspace{-.08cm} c_1 &(\bmod~{2n})\hspace{-.08cm}
	\end{array}\right\}.$$
\end{con}
The locating property of this code is given as follows.

\newcommand{\vpo}{\vphantom{1}}
\newcommand{\vpp}{\vphantom{\p}}
\newcommand*{\hl}{%
	\tcboxmath[on line, colback=black!5!, colframe=black!50!, arc=3pt, boxrule=0.8pt,left=1pt,right=0.5pt,top=0pt,bottom=-2pt]%
}
\newcommand*{\hla}{%
	\tcboxmath[colback=black!5!, colframe=black!50!, size=fbox, arc=3pt, boxrule=0.8pt,left=1pt,right=0.5pt,top=0pt,bottom=-13.5pt]%
}
\newcommand*{\hlb}{%
	\tcboxmath[colback=black!5!, colframe=black!50!, size=fbox, arc=3pt, boxrule=0.8pt,left=1pt,right=1pt,top=-2.5pt,bottom=-2pt]%
}

\begin{table}
	\setlength{\tabcolsep}{0.3pt}
	\centering
	\caption{Illustration of the different cases in the proof of Lemma \ref{lem::locating burst}. For each case, the original string $\x$ is depicted in the first row and the resulting string $\y$ in the second row. The below square brackets marks the position of the burst deletion and the gray background highlights the pattern $\p$. The position of the $(j-1)$-st pattern is marked by $S(j)$. Note that the table illustrates examples of each of the cases. There are some configurations, which are not captured by the table, e.g., it is possible in case 3 to delete the first part of the pattern. However, this does not affect the analysis of the proof.}
	{\renewcommand{\arraystretch}{2.1}
		\begin{tabular}{llc} \specialrule{.8pt}{0pt}{0pt}
			Case & \multicolumn{1}{c}{Burst type} & Remark \\ \specialrule{.8pt}{0pt}{0pt}
			\multirow{2.25}{*}{1.i} & $\phantom{\rightarrow~} (\dots\!\!\!\!\!\overset{\overset{S(j)}{\downarrow}}{\vpo} \!\!\!\! \hl{\p\vpo} ~  \x_1 ~ \underbracket{\x_2}_{k'} ~ \x_3~  \dots)$ & \multirow{2.25}{*}{$|\x_2| = k'$} \\[-.5ex]
			& $\rightarrow~(\dots\hl{\p\vpo} ~ \x_1 ~ \x_3 ~  \dots)$ & \\ \hline
			
			\multirow{2.25}{*}{1.ii} & $\phantom{\rightarrow~}(\dots \hl{\p\vpo} ~\x_1'~ 0\dots 0 ~ \hla{0\,. \underbracket{\vpp.\,.\,01_{}\,.\,.}_{k'}.\,1} ~ 1\dots 1 ~\x_2'~ \dots )$ & \\[1ex]
			& $\rightarrow~(\dots \hl{\vpo \p} ~\x_1'~ \hl{\vpp0\dots01\dots1}~\x_2'~ \dots )$ & \\ \hline
			
			\multirow{2.25}{*}{2} & $\phantom{\rightarrow~}(\dots \hl{\p\vpo} ~\x_1~\0^{k_0}~ \underbracket{\x_2}_{k'}\0^{k-k_0} ~\1^k ~ \dots)$ & \multirow{2.25}{*}{$|\x_2| = k'$} \\[-.5ex]
			 & $\rightarrow~(\dots\hl{\p\vpo}~\x_1~ \hlb{\vpp\0^k~\1^k} ~ \dots)$ & \\ \hline
			 
			\multirow{2.5}{*}{3} & $\phantom{\rightarrow~}(\dots \hl{\p\vpo} ~\x_1~\lefteqn{\hlb{\vpp\phantom{\0^k1\dots \,}}} ~\0^k1 \,.\,.\underbracket{\vpp.\, 1~~\x_2}_{k'} \x_3~\dots)$ & \\[-.5ex]
			& $\rightarrow~(\dots\hl{\p\vpo} ~\x_1~ \0^k\overbrace{1\dots1  }^{<k} ~\x_3~ \dots)$ & \\ \hline
			
			\multirow{2.75}{*}{4} & $\phantom{\rightarrow~}(\dots \hl{\p\vpo} ~\x_1  \lefteqn{\hlb{\phantom{\vpp\0^k1 \dots \,\, }}}~ \0^k1~.\,. \underbracket{\vpp.\,1 ~ \x_1  \lefteqn{\hlb{\phantom{\vpp0\dots 0^k }}} ~0\,.}_{k'} .\,.\,0 \1^k ~\dots )$ & \multirow{2.75}{*}{$|\x_1| \leq k'-2$} \\
			& $\rightarrow~(\dots \hl{\p\vpo} ~\x_1~ \0^k \overbrace{1\dots 1}^{\footnotesize < k} \overbrace{0\dots 0}^{\footnotesize < k} \1^k \dots)$ & \\ \specialrule{.8pt}{0pt}{0pt}
	\end{tabular}}
	\label{tab:cases}
\end{table}

\begin{lemma}\label{lem::locating burst}
	Let $\x \in \C_\mathsf{loc}(c_0,c_1)$ and $\y\in B_{\leq k}(\x)$. Given $\y,c_0,c_1$, it is possible to find in time $O(n)$ an interval $L \subseteq \mathbb{N}$ of length at most $\delta$, such that $\y = (\x_{[1,\ell]},\x_{[\ell+k'+1,n]})$ for some $\ell \in L$, where $k' = |\x|-|\y|$. 
\end{lemma}
\begin{proof}[Proof of Lemma~\ref{lem::locating burst}]
	We start with the observation that a burst of deletions of length $k'\le k$ can not destroy more than two patterns $\p$ or create more than one new pattern $\p$  in $\x$. Therefore, there are exactly four possible cases on the difference $n_{\p}(\x)-n_{\p}(\y)$, ranging from $-1$ to $2$. Moreover, this difference can be found by computing $c_0-n_{\p}(\y)~(\bmod ~4)$. Further, its possible to compute $\Delta' \mdef \Delta \pmod{2n}$ with $\Delta  \mdef \mathsf{VT}(\a_\p(\x))-\mathsf{VT}(\a_\p(\y))$ in $O(n)$ time. We distinguish between the above four cases, which are illustrated in Table \ref{tab:cases}. 
	\begin{enumerate}[wide, labelwidth=!]
		\item $n_{\p}(\x)-n_{\p}(\y) = 0$. In this case, we have two possibilities: i) no pattern is destroyed and no pattern is created, ii) one pattern is destroyed and one pattern is created. For both possibilities, the lengths of $\a_\p(\x)$ and $\a_\p(\y)$ agree and $\a_\p(\y)$ differs in at most two entries, $a_\p(\y)_j$  and possibly $a_\p(\y)_{j+1}$, from $\a_\p(\x)$ for some $1\leq j \leq n_\p(\x)+1$. Denoting \mbox{$k_1 \mdef \ a_\p(\x)_j - a_\p(\y)_j$} and $k_2 \mdef a_\p(\x)_{j+1} - a_\p(\y)_{j+1}$, we see that $k_1 + k_2 = k'$ and $k_1> 0$, $k_2\geq 0$. Note that possibility ii) can only occur when the burst deletes a part $(\0^{k_1}\1^{k_2})$ of the pattern $\p$. We obtain for the difference between the VT check sums of $\a_\p(\x)$  and $\a_\p(\y)$
		$$
		\Delta= jk_1 + (j+1)k_2=jk' + k_2.
		$$
		Since $0 < \Delta \leq n$, it follows that $\Delta'=\Delta$ and we can infer $j = \lfloor\frac{\Delta'}{k'}\rfloor$ and $k_2$ in $O(n)$ time. Define
		$$ S(j)\mdef \sum_{i<j} a_\p(\y)_i, $$
		in the following, which clearly can be computed, given $\y,c_0,c_1$, as $j$ can be inferred as discussed above. Note that $S(j)$ is exactly the position of the $(j-1)$-st pattern $\p$ in both $\x$ and $\y$. Its position is highlighted in Table \ref{tab:cases}. Moreover, using this definition, we can locate that the block of bits was deleted at a position within the range 
		$$
		L = \left[S(j)+1,S(j)+\delta\right]
		$$
		in the case i) and within the range
		$$ L = \left[ S(j+1)+1,S(j+1)+k \right] $$
		for the case ii), both of lengths at most $\delta$.
		\item $n_{\p}(\x)-n_{\p}(\y) =-1$. In this case, an additional pattern $\p$ is created in $\y$ and no other pattern is destroyed.  In other words $\a_\p(\y)$ is obtained from $\a_\p(\x)$ by replacing an element $a_\p(\x)_j$ for some $1\leq j \leq n_\p(\x)+1$ by $(a_\p(\y)_{j}, a_\p(\y)_{j+1})$, where $a_\p(\x)_j - a_\p(\y)_{j} - a_\p(\y)_{j+1}= k'$.
		Computing the difference between the VT check sums of $\a_\p(\x)$  and $\a_\p(\y)$ yields
		\begin{align*}
		\Delta&= ja_{\p}(\x)_j \!-\! \sum_{i> j}a_{\p}(\x)_i \!-\! ja_{\p}(\y)_j - (j+1)a_{\p}(\y)_{j+1}   \\
		&= j k'-\sum_{i> j}a_{\p}(\x)_i - a_\p(\y)_{j+1}= j k'-\sum_{i> j}a_{\p}(\y)_i.
		\end{align*}
		We note that the function $f_1(v)\mdef v k'-\sum_{i> v}a_{\p}(\y)_i \pmod{2n}$ is cyclically monotonically increasing in $v$ as both summands are increasing in $v$. $f_1(v)$ is also injective in $v$, since $k' \leq vk' \leq n$ and $0\leq \sum_{i> v}a_{\p}(\y)_i \leq n-k'$.  Therefore, and further since $f_1(v)$ and $\Delta$ are computable with the knowledge of only $\y,c_0,$ and $c_1$, we are able to find $j$ uniquely  in $O(n)$ time. Moreover, we can locate that the block of bits was deleted at a position within the range 
		$$
		L = \left[S(j)+1,S(j)+\delta\right]
		$$
		of length at most $\delta$.
		\item $n_\p(\x)-n_\p(\y)=1$. In this case, we have that one pattern $\p$ is destroyed in $\x$. Indeed, by the choice of the pattern $\p$ an instance when two patterns are destroyed and one is created is not possible. Therefore, $\a_\p(\y)$ is obtained from $\a_\p(\x)$ by replacing two elements $(a_\p(\x)_j, a_\p(\x)_{j+1})$ by  $a_\p(\y)_j$ with the property that $a_\p(\y)_j=a_\p(\x)_{j}+a_\p(\x)_{j+1}-k'$, where $1 \leq j \leq n_\p(\x)$. The difference between the VT check sums of $\a_\p(\x)$  and $\a_\p(\y)$ is
		\begin{align*}
		\Delta&=  \!\sum_{i> j+1} \hspace{-.08cm} a_{\p}(\x)_i \!+\! ja_{\p}(\x)_j \hspace{-.08cm}+\hspace{-.08cm} (j+1)a_{\p}(\x)_{j+1}  \!-\! ja_{\p}(\y)_j \\
		&= \sum_{i> j+1}a_{\p}(\x)_i +j k' + a_\p(\x)_{j+1}\\
		&= \sum_{i> j}a_{\p}(\y)_i +j k'+a_\p(\x)_{j+1}.
		\end{align*}
		It is important to recall that every component of the vector $\a_{\p}(\y)$ (except for possibly $a_{\p}(\y)_1$) is at least $2k$ and $k'\leq k$ and thus the function $f_2(v) \mdef \sum_{i> v}a_{\p}(\y)_i +v k' \pmod{2n}$ is again cyclically monotonically decreasing and injective in $v$. Assume first that $j=1$ and $a_{\p}(\x)_1 \leq 2k'$. This case is uniquely identifiable as only here $\Delta \geq n-k'$ and the converse is true for all other cases. Therefore, we can now assume that every component of $a_\p(\x)_i$ with $i\geq j$ is at least $2k\ge2k'$, which implies that $a_\p(\x)_{j+1}= a_\p(\y)_j-a_\p(\x)_{j}+k' \leq a_\p(\y)_j-k'$ and thus
		\begin{align*}
			\sum_{i> j}&a_{\p}(\y)_i +j k'< \Delta \leq \sum_{i> j-1}a_{\p}(\y)_i +(j-1) k' . 
		\end{align*}
		Thus, we can find $j$ uniquely  in $O(n)$ time by computing $f_2(v)$ for all $1\leq v \leq n_\p(\y)+1$ and choosing $j$, such that $f_2(j) < \Delta' \leq f_2(j-1)$, where the inequalities should be understood in a cyclic manner. Moreover, we can locate that the block of bits was deleted at a position within the range 
		$$
		L = \left\{ \begin{array}{ll}
		\left[S(j)+2k+1, S(j)+ 2k + \delta\right], & \text{ if } j \geq 2, \\
			\left[1, \delta\right]	, & \text{ if } j = 1
		\end{array} \right. ,
		$$
		which has length at most $\delta$.
		\item $n_\p(\x)-n_\p(\y)=2$. In this case, two patterns $\p$ are destroyed in $\x$. Therefore, $\a_\p(\y)$ is obtained from $\a_\p(\x)$ by replacing the triple $(a_\p(\x)_{j},a_\p(\x)_{j+1},a_\p(\x)_{j+2})$ for some $1\leq j \leq n_\p(\x)-1$ by one element $a_\p(\y)_j$, where $a_\p(\y)_j=a_\p(\x)_{j}+a_\p(\x)_{j+1}+a_\p(\x)_{j+2}-k'$ and $a_\p(\x)_{j+1}\leq 2k + k'-2$.
		The difference between the VT check sums of $\a_\p(\x)$  and $\a_\p(\y)$ is given by
		\begin{align*}
		\Delta&= 2\sum_{i> j+2}a_{\p}(\x)_i + ja_{\p}(\x)_j + (j+1)a_{\p}(\x)_{j+1} \\ & \hphantom{=} + (j+2)a_{\p}(\x)_{j+2}  - ja_{\p}(\y)_j \\
		&= 2\sum_{i> j+2}a_{\p}(\x)_i +j k' + a_\p(\x)_{j+1}+2a_\p(\x)_{j+2}\\
		&= 2\sum_{i> j}a_{\p}(\y)_i +j k'+a_\p(\x)_{j+1}+2a_\p(\x)_{j+2}.
		\end{align*}
		The function $f_3(v)\mdef 2\sum_{i> v}a_{\p}(\y)_i +v k' \pmod{2n}$ is cyclically monotonically decreasing and injective in $v$. Additionally, as for each $i\geq 2$, it holds true that $a_\p(\x)_i\geq 2k\ge2k'$, we have that $a_\p(\x)_{j+2} = a_\p(\y)_j-a_\p(\x)_{j}-a_\p(\x)_{j+1}+k' \leq a_\p(\y)_j-k'$ and hence
		\begin{align*}
			 2\sum_{i> j}&a_{\p}(\y)_i +j k' < \Delta < 2\sum_{i> j-1}a_{\p}(\y)_i +(j-1) k'.
		\end{align*}
		Thus, we can find $j$ uniquely in $O(n)$ time by choosing $j$, such that $f_3(j) < \Delta' \leq f_3(j-1)$, where the inequalities should be understood in a cyclic manner. Moreover, we can locate that the block of bits was deleted at a position within the range 
		$$
L = \left\{ \begin{array}{ll}
\left[S(j)+2k+1, S(j)+ 2k + \delta\right], & \text{ if } j \geq 2, \\
\left[1, \delta\right]	, & \text{ if } j = 1
\end{array} \right. ,
$$
		of length at most $\delta$.
	\end{enumerate}
	This completes the proof.
\end{proof}
\subsection{Shifted VT codes} \label{subsec:svt}
Having the knowledge of the approximate location of the burst of deletions, let us recall the concept of so called \emph{shifted VT codes}, introduced in~\cite{schoeny2017codes}. They are defined as follows.
\begin{con}
	Let $v,p,n \in \mathbb{N}$ with $0\leq v< p \leq n+1$ and $b \in \{0,1\}$. The \emph{shifted VT code} is defined by
	$$ \mathcal{C}^n_\mathsf{SVT}(v,b, p) = \left\{\x \in \{0,1\}^n : \begin{array}{ll}\mathsf{VT}(\x) = v &(\bmod~ p),\\ \mathsf{P}(\x) = b &(\bmod~ 2)\end{array} \right\}. $$
\end{con}
These codes are able to correct a single deletion once the position where the deletion
occurred is known to within an interval of size less than $p$. In particular, it is possible to prove the following property.
\begin{lemma}[Lemma 4 from~\cite{schoeny2017codes}]\label{lem::shifted vt code}
	Let $\x \in \mathcal{C}^n_\mathsf{SVT}(v,b, p)$ and $L=[j,j+p-2]$ for some $j$. Given any $\y=(\x_{[1:\ell-1]},\x_{[\ell+1:n]})$, where $\ell \in L$ and the knowledge of the set $L$, we are able to reconstruct $\x$ in $O(n)$ time.
\end{lemma}
\begin{proof}[Proof of Lemma~\ref{lem::shifted vt code}]
	Let $\c,\d \in \mathcal{C}^n_\mathsf{SVT}(v,b, p)$ be arbitrary with $\c\neq \d$ and $\y=(\c_{[1:\ell_1-1]},\c_{[\ell_1+1:n]})=(\d_{[1:\ell_2-1]},\d_{[\ell_2+1:n]})$, $\ell_2>\ell_1$ i.e., $\y$ is obtained from $\c$ by deleting $c_{\ell_1}$ and from $\d$ by deleting $d_{\ell_2}$. We will show that this is only possible when $\ell_2-\ell_1 \geq p$. 
	First, we observe that $c_{\ell_1}=d_{\ell_2}$ by 
	$$
	\mathsf{P}(\c)=\mathsf{P}(\d)=\mathsf{P}(\y)+c_{\ell_1} = \mathsf{P}(\y) + d_{\ell_2} \pmod{2}.
	$$
	The difference of the VT check sums of $\c$ and $\d$ is given by
	$$
	\mathsf{VT}(\c)-\mathsf{VT}(\d)= \ell_1 c_{\ell_1} - \ell_2 d_{\ell_2} + \sum_{i=\ell_1}^{\ell_2-1}y_i,
	$$
	and note that $\mathsf{VT}(\c)-\mathsf{VT}(\d) = 0 \pmod{p}$ by construction. If $c_{\ell_1}=d_{\ell_2}=0$, then $\mathsf{VT}(\c)-\mathsf{VT}(\d) = 0 \pmod{p}$ in only the case $\sum_{i=\ell_1}^{\ell_2-1}y_i=0 \pmod p$, which implies that $\ell_2-\ell_1\geq p$ or $y_i = 0$ for all $\ell_1\leq i<\ell_2$, which implies that $\c=\d$, which is a contradiction. Similarly, if $c_{\ell_1}=d_{\ell_2}=1$, then $\mathsf{VT}(\c)-\mathsf{VT}(\d) = 0 \pmod{p}$ implies that $\sum_{i=\ell_1}^{\ell_2-1}y_i=\ell_2-\ell_1 \pmod{p}$, which implies that either $\ell_2-\ell_1\geq p$ or $y_i=1$ for all $\ell_1\leq i <\ell_2$ and thus $\c=\d$, which is a contradiction. Therefore, since $\ell \in L$, $\x$ is the only possible codeword that could have resulted in $\y$, as for all other codewords $\x' \in \mathcal{C}^n_\mathsf{SVT}(v,b, p)$ with $\y = (\x'_{[1:\ell'-1]},\x'_{[\ell'+1:n]})$, $\ell' \notin L$.
\end{proof}
\subsection{Code construction}\label{ss::code construction}
We start by stating the final code construction, which is assembled using the locating code from Section \ref{subsec:locate} and the shifted VT codes discussed in the previous Section \ref{subsec:svt}.
\begin{con} \label{con:final}
	For arbitrary integers $c_0$, $c_1$, $\{v_{i,k'}\}_{1\le i \le k' \le k}$ and $\{b_{i,k'}\}_{1\le i \le k' \le k}$ we define a code $\C_{k}^n$ as follows
	$$
	\C_{k}^n=\left\{\x\in\{0,1\}^n:\,\,\begin{aligned}
	&\x\in \mathcal{C}_\mathsf{loc}(c_0,c_1),\\
	&\forall~i,k': 1\leq i\leq k' \leq k:\\
	&~~~ \x_{[i,n]_{k'}} \in \mathcal{C}_\mathsf{SVT}(v_{i,k'},b_{i,k'},\delta)
	\end{aligned}
	\right\}.
	$$
\end{con}
We prove the correctness of this construction and compute its redundancy in the following theorem.
\begin{theorem}\label{th::one-burst-deletion code}
	For any $n$, $c_0,c_1,\{v_{i,k'}\}_{1\le i \le k' \le k}$ and $\{b_{i,k'}\}_{1\le i \le k' \le k}$, the code $\C_{k}^n$ is a $(\leq k)$-burst-deletion-correcting code. Further, there exists a choice of these parameters such that the redundancy is at most
	$$n-\log|\C_{k}^n| \leq  \log n + {k+1 \choose  2} \log\log n + c_k,$$ for some constant $c_k$ that only depends on $k$. Moreover, any $\x\in \C_{k}^n$ can be recovered from any $\y\in B_{\le k}(\x)$ in $O(n)$ time.
\end{theorem}
\begin{proof}[Proof of Theorem~\ref{th::one-burst-deletion code}]
	We start by proving the upper bound on the redundancy. By Lemma~\ref{lem::cardinality of dense family}, we know that $|\{0,1\}^n\cap \D_{\p,\delta}|\ge 2^{n-1}$. As the number different code constructions is equal to the number of possibilities for variables $c_0\in[4],c_1\in[2n],\{v_{i,k'}\in[\delta]\}_{1\le i \le k' \le k}$ and $\{b_{i,k'}\in[2]\}_{1\le i \le k' \le k}$, we conclude that there is a code $\C_{k}^n$ with redundancy at most
	\begin{align*}
	n&-\log|\C_{k}^n| \leq n-\log\left(\frac{2^{n-1}}{8n 2^{{k+1 \choose 2}}\delta^{{k+1 \choose 2}}}\right)\\&=\log n + {k+1 \choose 2}\log \delta  +{k+1 \choose 2} + 4\\
	&\le\log n \!+\! {k\!+\!1 \choose 2}\log\log n\!+\! {k\!+\!1 \choose 2}(2k\!+\!2\!+\!\log k) \!+\!4\\
	&=\log n + {k+1 \choose 2}\log\log n + c_k.
	\end{align*}
	Let $\x\in\C_{k}^n$ and $\y\in B_{k'}(x)$ for some $k'\in[k]$. By Lemma~\ref{lem::locating burst}, we can locate in $O(n)$ time the position of the burst of deletions occurred up to the range of length $\delta$ consecutive positions. Additionally, we observe that for all $i$ with $1\leq i \leq k'$, $\y_{[i,n-k']_{k'}}$ is obtained from $\x_{[i,n]_{k'}}$ by deleting exactly one bit. Using the positional knowledge and Lemma~\ref{lem::shifted vt code}, we can reconstruct every $\x_{[i,n]_{k'}}$ in $O(n)$ time. Therefore, we can correct one burst of deletions and find $\x$ in $O(n)$ time. This completes the proof.
\end{proof} \vspace{-.23cm}
\begin{example}
	Exemplary code for $\delta=10$. $\mathcal{C}_2^{14} =$ \mbox{$ \{(0101\underline{0011}0\underline{0011}0),	(100\underline{0011}111\underline{0011}), (1001\underline{0011}1\underline{0011}1)\} $}. Note that there are larger codes of length $14$, however, here we present one of cardinality $3$ for reasons of  clarity.
\end{example}\vspace{-.23cm}
\section{Efficient encoding and decoding}\label{s::efficient construction}
Note that while Construction \ref{con:final} provides an efficiently decodable code that is able to correct a burst length at most $k$, it is not clear, how to efficiently encode in this code. In this section, we give a brief outline, how the previously introduced construction can be used to obtain an encoding algorithm that efficiently maps a string $\x\in\{0,1\}^n$ into a code that is able to correct a burst of length at most $k$. Note that here we only give the idea of the construction for brevity. The encoding procedure from an information word $\u \in \{0,1\}^d$ to a codeword $\x \in \{0,1\}^n$ works as follows. Define functions $\mathsf{E}(\u)$, and $\mathsf{S}_1(\u)$, where $\mathsf{E}(\u)$ is a function that maps a string to a $(\p,\delta)$-dense string, where $\delta = k2^{2k+1} \lceil\log d\rceil$ and $\mathsf{S}_1(\u)$ is a binary representation of the values $c_0$, $c_1$, $\{v_{i,k'}\}_{1\le i \le k' \le k}$ and $\{b_{i,k'}\}_{1\le i \le k' \le k}$ obtained by computing the syndromes from Construction \ref{con:final} of $\mathsf{E}(\u)$. Note that it is possible to find such functions that are efficiently computable, where $\mathsf{E}(\u)$ has redundancy that only depends on $k$ and $\mathsf{S}_1(\u)$ has redundancy similar to that derived in Theorem \ref{th::one-burst-deletion code}. With these functions, we define the following encoding map\vspace{-.1cm}
$$ \mathsf{Enc}(\u) =(\mathsf{E}(\u), \mathsf{E}(\mathsf{S}_1(\u)), R_{k+1}(\mathsf{S}_1(\mathsf{S}_1(\u)))), $$
where $R_{k+1}(\bullet)$ is the $(k+1)$-fold repetition code. It is straight-forward to derive that this encoding also introduces a redundancy of $\log n +\tilde{c}_k\log \log n$ for some constant $\tilde{c}_k$. The correctness of this construction should be understood by a decoding procedure that decodes from right to left. Let $\y \in B_{\leq k}(\mathsf{Enc}(\u))$ be received. First, it is possible to reconstruct $\mathsf{S}_1(\mathsf{S}_1(\u))$ by using the repetition code. Using $\mathsf{S}_1(\mathsf{S}_1(\u))$ and a possible erroneous version of $\mathsf{E}(\mathsf{S}_1(\u))$, we can reconstruct $\mathsf{S}_1(\u)$ and finally $\u$, the original information. 

	\bibliographystyle{IEEEtran}
	\bibliography{burst_deletion}

\end{document}